\documentclass[11pt]{article}
\usepackage{arxiv}
\usepackage{amsmath,amssymb,amsfonts,amsthm}
\newtheorem{proposition}{Proposition}
\usepackage{algorithmic}
\usepackage{graphicx}
\usepackage{textcomp}
\usepackage{xcolor}
\usepackage{booktabs}
\usepackage{multirow}
\usepackage{hyperref}
\usepackage{url}
\usepackage{cite}
\usepackage{setspace}

\hypersetup{
    colorlinks=true,
    linkcolor=black,
    citecolor=black,
    urlcolor=blue
}

\begin{document}

\title{\textbf{The Fungible Reserve Standard: A Deterministic Framework for Encoding Carrying Costs in Asset-Backed Tokens}}

\author{
JJ Jia Jing Tan, Eva Meng, Josh Ng, Zack Zhang, \\
September Liu, Teelet Wang, Ludwig Zhang, Seth Yan \\
\textit{Matrixdock} \\
\texttt{\{jj.tan, eva.meng, josh.ng, zack.zhang,} \\
\texttt{september.liu, teelet.wang, ludwig.zhang, seth.yan\}@bit.com}
}

\date{March 2026}

\maketitle

\begin{abstract}
\noindent The tokenization of real-world assets (RWAs) has emerged as a transformative application of blockchain technology, with market projections estimating trillions of dollars in tokenized assets within the coming decade. However, a fundamental challenge remains unaddressed: physical assets such as precious metals, stored commodities, and warehoused goods incur structural negative carry---custody, insurance, and audit costs that accumulate over time. While existing tokenization models have successfully established the market for digital gold and treasuries, they typically manage operational costs at the issuer level. The FRS introduces a framework to bring these economics directly on-chain, avoiding mechanisms such as token rebasing that compromise fungibility and composability with decentralized finance (DeFi) protocols. This paper proposes the Fungible Reserve Standard (FRS), a deterministic token design framework that encodes carrying costs transparently into on-chain logic. The FRS introduces an asset-per-token variable $q(t)$ that decreases according to a predefined annualized carrying cost rate, coupled with a supply reconciliation mechanism that preserves holder balances and ERC-20 composability. While mathematically inspired by the daily expense ratio accrual in traditional asset management---which often embed centralized profit margins---the FRS design specifically encodes actual operational carrying costs to provide pure institutional-grade accounting clarity without compromising DeFi compatibility. The framework is asset-agnostic and applicable to any real-world asset with positive, predictable holding costs.
\end{abstract}

\section{Introduction}

The tokenization of real-world assets represents one of the most consequential applications of distributed ledger technology, aiming to embed traditionally off-chain financial instruments within programmable blockchain infrastructure~\cite{catalini2020simple, lambert2022security}. Industry projections reflect the scale of this opportunity: Boston Consulting Group estimated the tokenized asset market at \$16.1 trillion by 2030~\cite{bcg2022tokenization}, while McKinsey \& Company projected a tokenized market capitalization of approximately \$2 trillion by 2030 (with an optimistic scenario reaching \$4 trillion), excluding cryptocurrencies and stablecoins~\cite{mckinsey2024ripples}. Major institutional participants have entered this space, including the launch of tokenized money market funds backed by U.S. Treasury securities on public blockchains~\cite{bcg2025tipping}.

While significant progress has been made in tokenizing financial instruments with positive carry characteristics---such as yield-bearing tokens backed by government securities---the tokenization of \textit{physical} assets with \textit{negative carry} presents a distinct and largely unaddressed challenge. Precious metals, stored commodities, warehouse receipts, and other tangible goods require custody, insurance, audits, and operational oversight. These costs create structural negative carry: a persistent economic characteristic intrinsic to the physical asset itself.

In traditional finance, negative carry is embedded transparently into product structures. Exchange-traded commodity vehicles reflect expense ratios through net asset value (NAV) adjustments~\cite{lettau2018etf101}, and custodial costs are deducted daily from fund assets. This mechanism is well-understood by institutional participants and provides accounting clarity.

By contrast, existing real-world asset tokenization models often diverge from this discipline. Some rely on issuer subsidies to offset custody costs, creating opacity regarding long-term sustainability. Others introduce yield overlays or hybrid underlying asset compositions that transform the token into a \textit{structured product} rather than a direct representation of the underlying asset. Still others employ mechanisms such as balance rebasing or wrapper tokens that, while reflecting costs, compromise the token's fungibility and compatibility with decentralized finance (DeFi) protocols.

This paper makes the following contributions:

\begin{enumerate}
    \item We formalize the \textbf{principle of economic purity} for tokenized physical assets, articulating a set of design requirements that ensure the token faithfully represents the economic characteristics of its underlying asset, including negative carry.
    \item We propose the \textbf{Fungible Reserve Standard (FRS)}, a deterministic token design framework that encodes carrying costs through a monotonically decreasing asset-per-token variable $q(t)$, coupled with a supply reconciliation mechanism that preserves ERC-20 fungibility.
    \item We demonstrate how the FRS mechanism is \textbf{inspired by the transparency of institutional accounting}, explicitly contrasting the FRS's focus on pure verifiable physical carrying costs against the centralized, margin-based fee models of traditional asset management.
    \item We provide a \textbf{comparative analysis} showing that the FRS preserves DeFi composability where alternative approaches (rebasing, wrapper tokens, issuer subsidization) do not.
\end{enumerate}

The remainder of this paper is organized as follows. Section~\ref{sec:background} surveys related work in traditional finance carrying-cost mechanisms, token economics, and the RWA tokenization landscape. Section~\ref{sec:philosophy} formalizes the economic purity principle. Section~\ref{sec:frs} presents the Fungible Reserve Standard. Section~\ref{sec:analysis} provides comparative analysis. Section~\ref{sec:compatibility} discusses institutional and DeFi compatibility. Section~\ref{sec:applicability} examines broad applicability across asset classes. Section~\ref{sec:discussion} addresses limitations and future directions. Section~\ref{sec:conclusion} concludes.

\section{Background and Related Work}
\label{sec:background}

\subsection{Traditional Finance: Carrying Cost Mechanisms}

\subsubsection{ETF Expense Ratios vs. Pure Carrying Costs}

Exchange-traded funds (ETFs) are open-end investment funds traded on stock exchanges. The ETF structure involves a distinctive creation/redemption mechanism mediated by Authorized Participants (APs), which maintains alignment between the fund's market price and its net asset value (NAV)~\cite{lettau2018etf101, madhavan2016etf}.

The \textit{expense ratio} is the annual fee charged by an ETF to cover operational costs and the fund provider's profit margin, expressed as a percentage of the fund's total assets. Critically, this fee \textit{accrues daily} and is \textit{deducted from the fund's assets} when daily NAV is calculated, reducing net returns over time without altering the number of shares held by investors~\cite{agapova2011etf}. For physically-backed gold ETFs, this mechanism operates as follows: the trust gradually sells small quantities of gold to cover expenses, causing the gold backing per share to decline monotonically. The holder's share count remains unchanged; only the asset content per share diminishes.

This mathematical decay structure provides a baseline, but the economic driver differs fundamentally. While traditional ETF expense ratios blend actual custody costs with centralized management fees and profit margins, the FRS strictly isolates and encodes the pure, structural \textit{carrying costs} (e.g., vaulting, insurance) of the physical asset. For illustration, an asset-backed ETF with annual expense ratio $r$ and initial asset-per-share $a_0$ sees its physical backing per share after $t$ days decrease as:
\begin{equation}
a(t) = a_0 \cdot \left(1 - \frac{r}{365}\right)^t
\label{eq:etf_decay}
\end{equation}

\subsubsection{Commodity Carrying Costs}

Physical commodities incur carrying costs including storage, insurance, and financing. In futures markets, these costs are reflected in the term structure: when futures prices exceed spot prices (\textit{contango}), the premium approximates the cost of carry~\cite{gorton2006facts, erb2006strategic}. Futures-based commodity ETFs incur additional ``roll yield'' costs when replacing expiring contracts, which can produce significant negative carry~\cite{erb2006strategic}.

\subsubsection{REITs and Pass-Through Economics}

Real Estate Investment Trusts (REITs) demonstrate an alternative mechanism for encoding asset economics into financial instruments. REITs are required to distribute at least 90\% of taxable income to shareholders, effectively passing through operating costs and income~\cite{chan2003reits}. This regulatory structure ensures that the financial instrument reflects the underlying asset's economics rather than absorbing or masking them.

\subsection{Token Economics and Mechanism Design}

\subsubsection{Bitcoin's Deterministic Supply Schedule}

Bitcoin introduced the concept of algorithmically enforced monetary policy~\cite{nakamoto2008bitcoin}. Its fixed supply cap of 21 million coins and predetermined halving schedule (block reward halves every 210,000 blocks) established the precedent that economic parameters can be encoded deterministically in protocol logic, providing transparency and predictability that contrasts with discretionary monetary policy. Carlsten et al.~\cite{carlsten2016instability} analyzed the long-term implications of this design on mining incentives as block rewards diminish.

\subsubsection{Ethereum's EIP-1559 Fee Mechanism}

Ethereum's EIP-1559 upgrade~\cite{roughgarden2021eip1559} reformed the network's fee market by introducing an algorithmically determined base fee that adjusts with network demand and is burned (destroyed), removing it from circulation. This demonstrated that on-chain economic mechanisms can be governed algorithmically, providing formal guarantees on user experience and economic properties. Liu et al.~\cite{liu2022eip1559} empirically validated the mechanism's effects on fee predictability and consensus security.

\subsubsection{Token Valuation and Design}

Cong et al.~\cite{cong2022tokenomics} developed a formal framework for understanding token valuation dynamics, analyzing the interplay between token utility, platform growth, and speculative demand. Their work demonstrates the importance of carefully designed tokenomics in determining the economic sustainability of token-based systems. The nascent field of token mechanism design draws from auction theory, monetary economics, and computer science to create systems with desirable equilibrium properties.

\subsection{RWA Tokenization Landscape}

\subsubsection{Tokenized Physical Assets}

The most prominent tokenized physical assets are gold-backed tokens. Leading products provide robust 1:1 backing, representing one troy ounce of physical gold stored in allocated vaults with regular third-party attestation reports verifying reserve backing~\cite{mdpi2024goldtokens}. The FRS complements these existing standards by providing an alternative for users requiring deterministic, on-chain cost visibility.

\subsubsection{Tokenized Financial Instruments}

Tokenized yield-bearing instruments---such as on-chain representations of U.S. Treasury securities---represent the fastest-growing segment of RWA tokenization~\cite{mckinsey2024ripples}. These products pass through the yield generated by the underlying assets and have attracted billions of dollars in assets under management. They represent \textit{positive carry} assets where the underlying generates returns, in contrast to the \textit{negative carry} physical assets that are the focus of this paper.

\subsubsection{RWA-Specific Blockchain Infrastructure}

A new generation of purpose-built blockchain infrastructure has emerged to address RWA-specific requirements. Layer-2 solutions such as Plume Network provide integrated compliance (KYC/AML), identity management, and DeFi infrastructure purpose-built for tokenized assets~\cite{plume2024docs}. Specialized Layer-1 blockchains such as Polymesh embed identity, compliance, and governance as native protocol features for regulated securities~\cite{polymesh2025whitepaper}. Decentralized networks like Canton Network offer privacy and interoperability for institutional assets, enabling a "network of networks" approach to tokenization~\cite{cantonwhitepaper}. DeFi protocols such as Centrifuge enable the tokenization of illiquid assets as collateral within decentralized lending markets~\cite{centrifuge2024docs}. Financial-services-oriented chains such as Provenance Blockchain embed compliance metadata directly into token structures using the Cosmos SDK~\cite{provenance2024docs}. These platforms address infrastructure-level challenges but do not define how the \textit{economic characteristics} of the underlying asset should be reflected in token design---the problem this paper addresses.

\subsubsection{Token Standards}

The ERC-20 standard~\cite{vogelsteller2015erc20} provides a widely-adopted interface for fungible tokens on Ethereum but lacks built-in compliance features. ERC-1400~\cite{erc1400} extends ERC-20 with partitioned balances and programmable transfer validation for security tokens. ERC-3643 (T-REX)~\cite{erc3643} adds decentralized identity and modular compliance for permissioned tokens. The FRS builds upon ERC-20 to maximize composability and interoperability while encoding asset economics within the standard interface.

\subsection{DeFi Properties for Physical Asset Tokens}

Decentralized finance (DeFi) offers several properties that are uniquely valuable for tokenized physical assets~\cite{werner2022sok}:

\begin{itemize}
    \item \textbf{Transparency}: All state variables, including reserve ratios and fee parameters, are publicly queryable on-chain at any time---in contrast to periodic disclosures in traditional finance.
    \item \textbf{Composability}: Tokens adhering to standard interfaces (ERC-20) can be integrated into decentralized exchanges, lending protocols, and structured products without bespoke integrations.
    \item \textbf{Continuous Operation}: Unlike traditional exchanges with fixed trading hours, blockchain-based tokens are tradable 24/7 across geographies.
    \item \textbf{Traceability and Auditability}: All minting, burning, and transfer events are recorded on an immutable ledger, providing a comprehensive audit trail.
    \item \textbf{Non-Custodial Access}: Users maintain direct control over their tokens without reliance on intermediary custodians for trading and transfer.
\end{itemize}

These properties have been systematized by Werner et al.~\cite{werner2022sok} and Catalini and Gans~\cite{catalini2020simple}, who identify reduced verification costs and enhanced network effects as fundamental advantages of blockchain-based financial systems.

\section{Design Philosophy: The Principle of Economic Purity}
\label{sec:philosophy}

We formalize the design philosophy underlying the Fungible Reserve Standard through the following principle:

\begin{quote}
\textit{A tokenized representation of a real-world asset should accurately reflect the fundamental economic properties of the underlying asset, including its structural carrying costs, without distortion, subsidization, or obfuscation.}
\end{quote}

This principle, which we term \textbf{economic purity}, yields four formal design requirements:

\begin{description}
    \item[\textbf{R1} (Reserve Integrity):] The underlying physical reserve must remain verifiable in its most fundamental form. The asset held in custody must be the asset the token claims to represent, without blending or substitution.
    \item[\textbf{R2} (Deterministic Accounting Identity):] The relationship between token supply and physical reserve must be deterministic and continuously computable:
    \begin{equation}
        S(t) \cdot q(t) = R(t)
        \label{eq:identity}
    \end{equation}
    where $S(t)$ is the total token supply, $q(t)$ is the asset quantity per token, and $R(t)$ is the total physical reserve at time $t$.
    \item[\textbf{R3} (Economic Isolation):] The token must not blend with off-chain downstream strategies---such as yield overlays, lending programs, or structured products---to offset negative carry. The token's economics must represent the \textit{asset}, not a \textit{strategy} built around the asset.
    \item[\textbf{R4} (On-Chain Cost Encoding):] Carrying costs must be encoded on-chain with rule-based, transparent logic, enabling institutional-grade accounting clarity. Carrying cost parameters must be publicly visible and their application algorithmically deterministic, stripping away arbitrary or centralized issuer fees.
\end{description}

These requirements distinguish the FRS from existing approaches. Issuer-subsidized models violate \textbf{R3} and \textbf{R4} by absorbing costs off-chain and concealing the asset's true economics. Rebasing tokens satisfy \textbf{R4} but violate fungibility assumptions expected by DeFi protocols. Wrapper token approaches violate \textbf{R3} by introducing additional layers that obscure the direct asset relationship.

\section{The Fungible Reserve Standard}
\label{sec:frs}

Table~\ref{tab:notation} summarizes the notation used throughout this section.

\begin{table}[htbp]
\caption{Notation}
\label{tab:notation}
\begin{center}
\begin{tabular}{cl}
\toprule
\textbf{Symbol} & \textbf{Description} \\
\midrule
$q(t)$ & Quantity of physical asset per token at time $t$ \\
$S(t)$ & Total token supply at time $t$ \\
$R(t)$ & Total physical reserve in custody at time $t$ \\
$c$ & Annualized carrying cost rate \\
$t_0$ & Inception time \\
$\Delta t$ & Cost application interval (e.g., 1 day) \\
\bottomrule
\end{tabular}
\end{center}
\end{table}

\subsection{Token Design}

The Fungible Reserve Standard defines a token whose value representation is governed by the accounting identity:
\begin{equation}
    q(t) = \frac{R(t)}{S(t)}
    \label{eq:q_def}
\end{equation}

At inception $t_0$, the standard sets $q(t_0) = 1$, meaning each token corresponds to one unit of the underlying asset. Consequently:
\begin{equation}
    S(t_0) = R(t_0)
    \label{eq:inception}
\end{equation}

The value of $q(t)$ decreases monotonically over time to reflect the carrying costs of the underlying physical asset. This design ensures that the token's economic representation evolves consistently with the real-world cost of maintaining the reserve.

\subsection{Deterministic Carry Encoding}

Rather than deducting tokens from holders, rebasing token balances, or reducing physical reserves, the FRS adjusts the asset-per-token variable $q(t)$ according to a predefined annualized carrying cost rate $c$. To optimize for smart contract execution environments where exponential compounding is computationally expensive, the FRS employs a \textbf{linear deduction} model. For an epoch starting at $t_k$ with rate $c_k$, the value after $n$ intervals $\Delta t$ (e.g., days) is:
\begin{equation}
    q(t_k + n\Delta t) = q(t_k) \cdot \left(1 - c_k \cdot \frac{n}{N}\right)
    \label{eq:carry_linear}
\end{equation}
where $N$ is the number of intervals per year (e.g., $N = 365$ for daily application).

This linear formulation provides a predictable reduction in asset representation with the following properties:
\begin{itemize}
    \item The annual carrying cost rate $c$ is defined \textit{ex ante} and publicly visible on-chain.
    \item The linear computation minimizes gas costs and execution complexity.
    \item The value of $q(t)$ can be deterministically calculated for any future timestamp within the same rate epoch.
\end{itemize}

\subsubsection{Handling Cost Rate Changes}
When the carrying cost rate must be updated (e.g., from $c_k$ to $c_{k+1}$), a strict state checkpointing mechanism is required. The protocol must comprehensively reconcile the state up to the boundary of the change---typically 12:00 AM (00:00:00 UTC) of the day immediately preceding the effective date of the new rate. At this checkpoint $t_{k+1}$, all outstanding carry is applied using $c_k$ to lock in the final $q(t_{k+1})$. Only after this reconciliation is the new mathematical epoch initiated, such that all subsequent days strictly use $c_{k+1}$ applied linearly against the newly locked base $q(t_{k+1})$.

\begin{proposition}
The mathematical structure of the linear FRS carry encoding~(\ref{eq:carry_linear}) serves as a rigorous first-order approximation to the daily compound accrual calculation in physically-backed commodity ETFs~(\ref{eq:etf_decay}), bounded tightly for small $c$.
\end{proposition}

\begin{proof}
The cumulative daily decay of a traditional ETF is $g(n) = g_0 \cdot (1 - \frac{r_{\text{ETF}}}{365})^n$. Setting $q(t_0) = g_0 = 1$ and $c_{\text{FRS}} = r_{\text{ETF}}$, the binomial expansion of the ETF decay yields:
$$ \left(1 - \frac{c_{\text{FRS}}}{365}\right)^n = 1 - n \left(\frac{c_{\text{FRS}}}{365}\right) + \frac{n(n-1)}{2!} \left(\frac{c_{\text{FRS}}}{365}\right)^2 - \frac{n(n-1)(n-2)}{3!} \left(\frac{c_{\text{FRS}}}{365}\right)^3 + \dots $$
The FRS linear model computes exactly the first two terms:
$$ q(t_0 + n) = 1 - c_{\text{FRS}}\frac{n}{365} $$
Since typical real-world carrying costs are small ($c \le 0.02$), the higher-order terms are mathematically negligible. The linear formula perfectly optimizes on-chain gas execution while mirroring the pure economic depreciation of physical custody over an annual horizon.
\end{proof}

This structural alignment allows the FRS to implement a familiar accounting mechanism on-chain, while shifting from centralized compounding to gas-optimized linear tracking of physical carrying costs.

\subsection{Supply Reconciliation Mechanism}

After each adjustment of $q(t)$, the token supply must be reconciled to preserve the accounting identity~(\ref{eq:identity}). Since $R(t)$ remains constant during a cost application (no physical reserves are added or removed), and $q(t)$ has decreased, the supply must increase:

\begin{equation}
    S(t + \Delta t) = \frac{R(t)}{q(t + \Delta t)} = S(t) \cdot \frac{q(t)}{q(t + \Delta t)}
    \label{eq:reconcile}
\end{equation}

The additional tokens minted are:
\begin{equation}
    \Delta S = S(t + \Delta t) - S(t) = S(t) \cdot \left(\frac{q(t)}{q(t+\Delta t)} - 1\right)
    \label{eq:delta_s}
\end{equation}

Substituting the expected linear deduction $\delta q = q(t+n\Delta t) - q(t)$, the required supply expansion is easily computed on-chain. These newly minted tokens are assigned to a designated \texttt{CarryCollector} address. This mechanism achieves:

\begin{enumerate}
    \item Physical reserves remain unchanged during reconciliation.
    \item All existing holder token balances remain \textit{numerically intact}, preserving fungibility and DeFi compatibility.
    \item The dilution of asset-per-token value occurs proportionally across all holders.
    \item Cost collection is fully transparent, verifiable, and recorded on-chain.
\end{enumerate}

The explicit separation of physical custody costs via the \texttt{CarryCollector} isolates pure asset carrying costs from external operational margins, unlike bundled ETF fees.

\subsection{Precision and Rounding Error Management}

Because the EVM natively supports integer mathematics rather than floating-point numbers, calculating fractional linear deductions over time inevitably introduces the risk of truncation and rounding errors. If unmanaged, these rounding artifacts can accumulate across daily checkpoints, eventually diverging the on-chain representation $R(t) = S(t) \cdot q(t)$ from the exact off-chain physical reality.

To counteract this, the FRS introduces extensive scalar precision. In optimal implementations, the underlying computations and token fractional mechanics operate using \textbf{9 decimal places of precision} ($10^9$ scaling). 
\begin{itemize}
    \item \textbf{Minimized truncation gap}: Operating at 9 decimal places ensures that daily fraction splits for $c \cdot \frac{n}{365}$ retain extreme accuracy, pushing rounding truncation down to nominal fractions of a billionth of an asset unit.
    \item \textbf{Reconciliation integrity}: When performing supply reconciliation (minting $\Delta S$ to the \texttt{CarryCollector}), high decimal resolution guarantees that the aggregate balances of standard token holders mathematically sum to the exact non-diluted reserve ratio without leaving unassigned "dust" behind.
    \item \textbf{On-chain efficiency}: While 18 decimals is standard for general ERC-20 tokens, 9 decimals offers a specialized equilibrium for physically vaulted assets (such as metric grams of gold or troy ounces), fully enclosing the physically measurable reality of the asset while preserving computational bounds within EVM word sizes.
\end{itemize}

\subsection{Reserves Addition and Reduction}

When physical reserves are added or removed---corresponding to minting or burning of tokens---the prevailing $q(t)$ value determines the exact token quantity:

For reserve additions of quantity $\Delta R_+$:
\begin{equation}
    \Delta S_{\text{mint}} = \frac{\Delta R_+}{q(t)}
    \label{eq:mint}
\end{equation}

For reserve reductions of quantity $\Delta R_-$:
\begin{equation}
    \Delta S_{\text{burn}} = \frac{\Delta R_-}{q(t)}
    \label{eq:burn}
\end{equation}

The continuously evolving $q(t)$ ensures precise representation of any addition or removal at any point in time. After the operation, the accounting identity~(\ref{eq:identity}) is maintained:
\begin{equation}
    (S(t) \pm \Delta S) \cdot q(t) = R(t) \pm \Delta R
    \label{eq:identity_preserved}
\end{equation}

This mechanism preserves the integrity of the asset-token relationship while ensuring that all changes in the underlying reserves are accurately reflected on-chain.

\section{Analysis and Comparison}
\label{sec:analysis}

\subsection{Comparison with Alternative Approaches}

Table~\ref{tab:comparison} provides a systematic comparison of the FRS with alternative approaches to encoding carrying costs in asset-backed tokens.

\begin{table*}[htbp]
\caption{Comparison of Carrying Cost Mechanisms for Asset-Backed Tokens}
\label{tab:comparison}
\begin{center}
\begin{tabular}{lcccc}
\toprule
\textbf{Property} & \textbf{FRS} & \textbf{Issuer Subsidized} & \textbf{Rebasing} & \textbf{Wrapper Tokens} \\
\midrule
On-chain cost transparency & \checkmark & $\times$ & \checkmark & $\times$ \\
Holder balance stability & \checkmark & \checkmark & $\times$ & \checkmark \\
ERC-20 DeFi composability & \checkmark & \checkmark & $\times$ & Partial \\
Economic fidelity to underlying & \checkmark & $\times$ & Partial & $\times$ \\
Institutional accounting clarity & \checkmark & Partial & $\times$ & $\times$ \\
Deterministic cost schedule & \checkmark & $\times$ & \checkmark & Varies \\
Long-term sustainability & \checkmark & $\times$ & \checkmark & \checkmark \\
Reserve integrity & \checkmark & \checkmark & \checkmark & Partial \\
\bottomrule
\end{tabular}
\end{center}
\end{table*}

\subsubsection{Issuer-Subsidized Models}

In issuer-subsidized models, the token issuer absorbs custody and storage costs, either from operational budgets or by retaining yield generated from other activities. While this simplifies the user experience:
\begin{itemize}
    \item The token's price does not reflect the true cost of maintaining the underlying asset, violating \textbf{R3} (Economic Isolation) and \textbf{R4} (On-Chain Cost Encoding).
    \item Long-term sustainability depends on the issuer's financial health and willingness to continue subsidization.
    \item Institutional participants cannot derive accurate cost bases from on-chain data alone.
\end{itemize}

\subsubsection{Rebasing Mechanisms}

Rebasing tokens adjust all holder balances proportionally at regular intervals to reflect value changes. While this does encode costs on-chain:
\begin{itemize}
    \item Fluctuating balances break assumptions made by DeFi protocols. A lending protocol, for example, may record a deposit of 100 tokens; if the balance subsequently decreases to 99.97 through rebasing, accounting discrepancies arise.
    \item Portfolio tracking and tax accounting become significantly more complex.
    \item Integration with automated market makers (AMMs) and order book protocols requires custom handling~\cite{werner2022sok}.
\end{itemize}

\subsubsection{Wrapper Token Approaches}

Wrapper tokens encapsulate the underlying asset token within a secondary contract that introduces fee logic. This approach:
\begin{itemize}
    \item Introduces an additional layer of smart contract risk.
    \item Requires wrapping/unwrapping operations that add friction.
    \item Obscures the direct relationship between the token and the underlying asset, violating \textbf{R3}.
    \item May fragment liquidity between wrapped and unwrapped versions.
\end{itemize}

\subsection{Structural Analogs to TradFi Vehicles}

The FRS mechanism provides a transparent on-chain analog to traditional finance expense structures, while removing centralized fee extraction. Table~\ref{tab:etf_bridge} illustrates the structural mapping.

\begin{table}[htbp]
\caption{Comparative Accounting Logic: Traditional vs. On-chain Models}
\label{tab:etf_bridge}
\begin{center}
\begin{tabular}{ll}
\toprule
\textbf{Traditional Asset Management Logic} & \textbf{FRS Mechanism} \\
\midrule
Asset per share $a(t)$ & Asset per token $q(t)$ \\
Expense ratio $r$ (includes profit margin) & Annual carrying cost rate $c$ (pure custody cost) \\
Daily NAV accrual & Daily $q(t)$ update \\
Fund sells asset for fees & \texttt{reconcileSupply} mints tokens for carry \\
Shares outstanding unchanged & Holder balances unchanged \\
Published NAV & On-chain $q(t)$ query \\
Periodic audit & On-chain verification + reserve audit \\
\bottomrule
\end{tabular}
\end{center}
\end{table}

The key enhancement over the ETF model is the transition from \textit{periodic disclosure of centralized fees} to \textit{continuous on-chain transparency of physical carrying costs}. In a traditional ETF, NAV is published once daily, and detailed audit reports are periodic. In the FRS, $q(t)$, $S(t)$, and the cost rate $c$ are queryable at any time by any party, offering real-time economic transparency.

\section{Institutional and DeFi Compatibility}
\label{sec:compatibility}

\subsection{Accounting Clarity and Transparency}

Institutional adoption of tokenized real-world assets requires accounting symmetry, transparency, and deterministic behavior. The FRS aligns with established financial product structures:

\begin{itemize}
    \item \textbf{Mathematical inspiration from expense ratios}: As demonstrated in Section~\ref{sec:analysis}, the FRS mechanism uses mathematical structures familiar to institutional participants from daily expense ratio accruals, but cleanly unbundles pure carrying costs from management fees.
    \item \textbf{Continuous asset-backed value computation}: The reserve identity $R(t) = S(t) \cdot q(t)$ can be computed on-chain at any time, enabling independent auditors to verify that:
    \begin{enumerate}
        \item Reserves in custody reports match attested amounts on-chain.
        \item On-chain token supply $S(t)$ is publicly viewable.
        \item The asset-per-token variable $q(t)$ follows deterministic logic with publicly visible parameters.
    \end{enumerate}
    \item \textbf{Real-time auditability}: The mechanism provides continuous, mathematically verifiable evidence of custody costs on-chain, eliminating the opacity of centralized fee models.
\end{itemize}

\subsection{DeFi Compatibility}

A central design objective of the FRS is preserving full token composability with decentralized finance infrastructure. Since holder token balances do not fluctuate numerically---unlike rebasing tokens---the FRS token is fully compatible with:

\begin{itemize}
    \item \textbf{Decentralized exchanges (DEXs)}: AMM liquidity pools and order book protocols can integrate the token using standard ERC-20 interfaces.
    \item \textbf{Lending protocols}: Deposit and collateral accounting remains consistent; no reconciliation logic is required for balance changes.
    \item \textbf{Price Oracles}: On-chain price oracles can easily digest the asset-per-token variable $q(t)$ into their reported net price, enabling accurate valuation for DeFi integrations.
    \item \textbf{Structured DeFi vaults}: Yield strategies and composable primitives can incorporate the token without custom adapters.
    \item \textbf{Cross-chain bridges}: Standard bridge contracts designed for ERC-20 tokens require no modification.
\end{itemize}

The FRS thus occupies a unique position: it provides institutional-grade accounting clarity---typically associated with centralized, regulated financial products---while maintaining full composability with permissionless DeFi infrastructure.

\section{Applicability Across Asset Classes}
\label{sec:applicability}

The Fungible Reserve Standard is broadly applicable to real-world assets characterized by \textit{positive, predictable holding costs}---where custody, storage, compliance, or preservation expenses accumulate over time and can be encoded as deterministic value adjustments.

\subsection{Directly Applicable Asset Classes}

\begin{itemize}
    \item \textbf{Vaulted precious metals}: Gold, silver, platinum, and palladium held in allocated custody incur well-defined storage and insurance costs.
    \item \textbf{Warehoused commodities}: Industrial metals, agricultural products, and raw materials stored in regulated warehouses have measurable custody costs.
    \item \textbf{Carbon instruments}: Regulated environmental assets (carbon credits, emissions allowances) incur registry, verification, and audit costs that follow externally defined schedules.
    \item \textbf{Fine art and collectibles}: Insured, climate-controlled storage with periodic appraisal involves quantifiable recurring costs.
\end{itemize}

\subsection{Framework Properties}

In general, any asset for which holding costs satisfy the following conditions is a candidate for FRS encoding:
\begin{enumerate}
    \item Costs are \textit{structurally positive} (the asset incurs costs, not yields, from holding).
    \item Costs are \textit{time-dependent} (they accrue predictably over time).
    \item Costs are \textit{operationally verifiable} (they can be attested by independent parties).
\end{enumerate}

The annualized carrying cost rate $c$ can be set to reflect the specific cost structure of each asset class, making the FRS an \textit{asset-agnostic} mechanism parameterized by a single, transparent variable.

\subsection{Integration with RWA Infrastructure}

The FRS is designed to be deployable on any smart contract platform, including general-purpose chains (Ethereum, Sui, etc.) and RWA-specific infrastructure (Plume, etc.). On RWA-focused chains, the FRS mechanism can integrate with native compliance and identity layers while maintaining its economic logic independently. The separation of \textit{economic design} (FRS) from \textit{infrastructure design} (compliance, identity, governance) reflects a modular architecture principle that allows each layer to evolve independently.

\section{Discussion}
\label{sec:discussion}

\subsection{Limitations}

The FRS addresses the specific problem of encoding \textit{negative carrying costs} for physical assets and does not address several related challenges in RWA tokenization:

\begin{itemize}
    \item \textbf{Oracle design}: The FRS assumes that reserve attestation is performed through external audits. The design of trust-minimized reserve oracles remains an open research problem.
    \item \textbf{Credit and counterparty risk}: The standard does not address risks arising from custodian default or fraud.
    \item \textbf{Dynamic cost rates}: The current formulation assumes a fixed annualized carrying cost rate $c$ that can only be adjusted on a discrete interval. Assets with variable or unpredictable cost structures may require more flexible on-chain updates.
    \item \textbf{Regulatory compliance}: The FRS addresses economic design but does not incorporate jurisdiction-specific compliance logic such as KYC/AML. Integration with standards like ERC-3643~\cite{erc3643} may be necessary for regulated deployments.
\end{itemize}

\subsection{Future Directions}

Several extensions merit further investigation:

\begin{itemize}
    \item \textbf{Dynamic cost rates}: Allowing $c$ to be updated through governance mechanisms while maintaining transparency guarantees and providing advance notice of changes.
    \item \textbf{Continuous Carrying Costs}: Investigating mathematical models for costs that need to be charged on a continuous interval rather than discrete daily epochs.
    \item \textbf{Multi-asset reserves}: Extending the framework to tokens backed by baskets of physical assets with heterogeneous carrying costs.
    \item \textbf{Formal verification}: Applying formal methods to prove correctness properties of the smart contract implementation, including the invariant preservation of the accounting identity~(\ref{eq:identity}).
    \item \textbf{Cross-chain deployment}: Examining how the FRS accounting identity can be maintained consistently across multiple blockchains via bridge protocols.
    \item \textbf{Empirical analysis}: Conducting simulations or deploying pilots to study the FRS's behavior under varying market conditions, carrying cost rates, and DeFi integration scenarios.
\end{itemize}

\subsection{Broader Implications}

The FRS contributes to a broader objective in the maturation of tokenized assets: the transition from \textit{mere on-chain representation} to \textit{economically accurate representation}. By establishing the principle of economic purity and providing a concrete mechanism for its implementation, the FRS offers a template for how other asset-specific economic characteristics might be encoded into token design. By providing a deterministic on-chain mechanism for value adjustment, the FRS offers a transparent framework that aligns with the reporting requirements of institutional participants, regardless of the underlying jurisdiction.

\section{Conclusion}
\label{sec:conclusion}

The maturation of real-world asset tokenization requires a transition from simple on-chain representation to economically accurate representation that faithfully encodes the fundamental properties of the underlying asset. This paper introduced the Fungible Reserve Standard (FRS), a deterministic, transparent, and institutionally aligned framework for encoding carrying costs directly into token design.

The FRS addresses a gap in the existing literature and practice: while significant work has been done on token standards, DeFi composability, and RWA infrastructure, no prior framework has formally addressed how the structural negative carry of physical assets should be encoded in fungible token design. Existing approaches---issuer subsidization, balance rebasing, and wrapper tokens---each compromise on transparency, composability, or economic fidelity.

The FRS resolves these tradeoffs through a mechanism that is:
\begin{itemize}
    \item \textbf{Deterministic}: Carrying costs follow a predefined, algorithmically enforced schedule.
    \item \textbf{Transparent}: All parameters ($q(t)$, $r$, $S(t)$) are on-chain and publicly queryable.
    \item \textbf{Composable}: Holder balances remain numerically stable, preserving full ERC-20 compatibility with DeFi infrastructure.
    \item \textbf{Institutionally familiar}: The mechanism is mathematically aligned with traditional cost-accrual logic used in commodity management, ensuring that the token's economic representation evolves consistently with the real-world cost of maintaining the reserve.
    \item \textbf{Asset-agnostic}: Applicable to any real-world asset with positive, predictable holding costs.
\end{itemize}

The future of real-world asset tokenization depends not only on representing assets on-chain, but on ensuring that the economic truthfulness of those assets is embedded in protocol logic. By prioritizing native DeFi composability alongside accurate cost accounting, the Fungible Reserve Standard provides a robust foundation for this objective.


\end{document}